\documentclass{article}

\usepackage{cite}

\usepackage{spconf,epsfig}
\usepackage{graphicx}
\usepackage{tikz}

\usepackage[cmex10]{amsmath}
\usepackage{amsthm}
\usepackage{amssymb}
\usepackage{enumitem}
\usepackage{balance}
\usepackage{array}
\usepackage{bm}
\usepackage{color}

\newtheorem{thm}{Theorem}[section]

\newtheorem{lemma}[thm]{Lemma}

\theoremstyle{definition}
\newtheorem{defin}[thm]{Definition}

\theoremstyle{remark}
\newtheorem{remark}[thm]{Remark}

\theoremstyle{remark}

\DeclareMathOperator*{\rank}{rank}

\def\R {{\Bbb R}}

\newcommand{\hsk}{\hskip 0.3cm}

\newcommand{\nl}{\newline}
\newcommand{\ben}{\begin{enumerate}}

\newcommand{\een}{\end{enumerate}}
\newcommand{\bit}{\begin{itemize}}

\newcommand{\eit}{\end{itemize}}

\def\eps{\varepsilon}

\def\dispsum{\displaystyle \sum}
\def\meansum{\frac{1}{N}\dispsum_{n=1}^{N}}

\def\QED{\nobreak\quad\ifmmode\roman{Q.E.D.}\else{\rm Q.E.D.}\fi}

\newcommand{\abs}[1]{\left\lvert#1\right\rvert}
\def\norm#1{\left\Vert#1\right\Vert}
\def\BracketRef#1{(\ref{#1})}
\def\mbf#1{\bm{#1}}

\begin{document}

\title{Finite sample performance of linear least squares estimators under sub-Gaussian martingale difference noise}

\ninept


\twoauthors
  {Michael Krikheli \sthanks{This work is part of the first author's Ph.D. thesis. This work is partially supported by ISF grant 903/2013.}}
	{Faculty of Engineering, \\
    Bar-Ilan University 52900, \\
    Ramat-Gan Israel, \\
    michael.krih@gmail.com}
  {Amir Leshem}
	{Faculty of Engineering, \\
    Bar-Ilan University 52900, \\
    Ramat-Gan Israel}
 
\maketitle
\begin{abstract}
Linear Least Squares is a very well known technique for parameter estimation, which is used even when sub-optimal, because of its very low computational requirements and the fact that exact knowledge of the noise statistics is not required. Surprisingly, bounding the probability of large errors with finitely many samples has been left open, especially when dealing with correlated noise with unknown covariance. In this paper we analyze the finite sample performance of the linear least squares estimator under sub-Gaussian martingale difference noise. In order to analyze this important question we used concentration of measure bounds. When applying these bounds we obtained tight bounds on the tail of the estimator's distribution. We show the fast exponential convergence of the number of samples required to ensure a given accuracy with high probability. We provide probability tail bounds on the estimation error's norm. Our analysis method is simple and uses simple $L_{\infty}$ type bounds on the estimation error. The tightness of the bounds is tested through simulation. The proposed bounds make it possible to predict the number of samples required for least squares estimation even when least squares is sub-optimal and used for computational simplicity. The finite sample analysis of least squares models with this general noise model is novel.
\end{abstract}
\section{Introduction}
\subsection{Related Work}
Linear least squares estimation has numerous applications in many fields. For instance, it was used in soft-decision image interpolation applications in ~\cite{zhang2008image} and ~\cite{hung2012robust}. Another field that uses linear least squares is source localization using signal strength, as in ~\cite{so2011linear}. In that paper, weighted linear least squares was used to find the distance of the received signals given the strength of the signals received in the sensors and the sensors' locations. Weighted least squares estimators were also used in the field of diffusion MRI parameters estimation ~\cite{veraart2013weighted}. It was shown that the weighted linear least squares approach has significant advantages because of its simplicity and good results. A standard analysis of estimation problems calculates the Cramer-Rao bound (CRB) and uses the asymptotic normality of the estimator. This type of analysis is asymptotic by nature. For some applications, see for instance ~\cite{leshem1999direction} where direction of arrival problems were analyzed in terms of the CRB. In ~\cite{stoica1989music} the ML estimator and MUSIC algorithms were studied and the CRB was calculated. However, as is well known, the Central Limit Theorem, and the Gaussian approximation are not valid in the case of rare large errors. In many applications the performance is not impacted by small errors, but large errors can lead to catastrophic results. One such example is in wireless communication channel estimation, where the accuracy of the channel estimation should suffice for the given modulation. However rare events where the estimation is significantly far away can lead to total failure. Furthermore, in such applications, training is short and we cannot rely on asymptotic large deviation results. Hence we need tight upper bounds on the $L_{\infty}$ norm of the error. \nl
The noise model differs accross applications of least squares and other optimization methods. Rather than the Gaussian model a Gaussian mixture is used in many applications. For instance, in ~\cite{djuric2002sequential} a Gaussian mixture model of a time-varying autoregressive process was assumed and analyzed. The Gaussian mixture model was used to model noise in underwater communication systems in ~\cite{banerjee2014performance}. Wiener filters in Guassian mixture signal estimation were analyzed in ~\cite{tan2014wiener}. In ~\cite{bhatia2007non} a likelihood based algorithm for Gaussian mixture noise was devised and analyzed in the terms of the CRLB. In ~\cite{wang1999robust} a robust detection technique using Maximum-Likelihood estimation was proposed for an impulsive noise modeled as a Gaussian mixture. In this work we consider sub-Gaussian noise, which is a general non-Gaussian noise framework. The Gaussian mixture model, for instance, is sub-Gaussian and our results are valid for this model. In the case of Gaussian noise, least squares coincides with the maximum likelihood estimator. Still, in  many cases of interest least squares estimation is used in non-Gaussian noise as well for computational simplicity. Specifically the sub-Gaussian noise model is of special interest in many applications.\nl
In many cases the noise model used is not i.i.d but the noise is correlated. An important case is that of martingale difference noise. This noise model is quite general and is used in various fields. For example the first order ARCH models introduced in ~\cite{engle1982autoregressive} are popular in economic theory. Moreover, ~\cite{lai1982least} analyzed similar least squares models with applications in control theory. The asymptotic properties of these models have been analyzed in various papers, for example ~\cite{lai1983asymptotic, nelson1980note, christopeit1980strong}. The results show the strong consistency of the least squares estimator under martingale difference noise and for autoregressive models. The least squares efficiency in an autoregressive noise model was studied in ~\cite{kramer1980finite}. However, finite sample results were not given.\nl
The least squares problem is well studied. The strong consistency of the linear least squares was proved in ~\cite{lai1985strong}. Asymptotic bounds for fixed size confidence bounds were stated for example in ~\cite{gleser1965asymptotic}.
In the past few years, the finite sample behavior of least squares problems has been studied in ~\cite{oliveira2013lower,hsu2014random,audibert2010robust,audibert2011robust}. Some of these results also analyze regularized least squares models. These results only studied the i.i.d noise case. In this work we extend these results to the sub-Gaussian MDS noise case which is much more general. Beyond the theoretical results we also provide simulated examples of the bounds for the problem of channel estimation with a random mixing matrix.\nl

\subsection{Contribution}
In this paper we provide a finite sample analysis of linear least squares problems under sub-Gaussian martingale difference sequence (MDS) noise.
We provide $L_{\infty}$ error bounds that can be used to compute the confidence interval in a non-parametric way (i.e., without knowing the exact distribution) of the estimation error. The main theorem of this paper allows us to compute the performance of linear least squares under very general conditions. Since the linear least squares solution is computationally simple it is used in practice even when it is sub-optimal. The analysis of this paper allows the designer to understand the loss due to the computational complexity reduction without the need for massive simulations. We extend the results of ~\cite{krikheli2016finite} in two significant ways. The first is allowing the mixing matrix to be a general bounded elements matrix. More importantly, we extend the analysis to the case of sub-Gaussian MDS noise. The sub Gaussian martingale noise covers many examples of correlated noise, and specifically the case of an interfering zero mean signal which passes through a finite impulse response channel. Hence we are able to predict large error behavior. This provides finite sample analysis under a very general noise framework. While the bounds are not tight, they are still useful and pave the way to further analyses which may tighten these bounds even further. The fact that we only need knowledge of a sub-Gaussianity parameter of the noise allows us to use these bounds when the noise distribution is unknown. 
\section{problem formulation}\label{sec:problem_formulation}
Consider a linear model with additive noise
\begin{equation}\label{eq:x_definition}
\mbf{x} = \mbf{A}\mbf{\theta}_0 + \mbf{v}
\end{equation}
where $\mbf{x} \in \R^{N\times 1}$ is our output, $\mbf{A} \in \R^{N\times p}$ is a known matrix with bounded random elements, $\mbf{\theta}_0 \in \R^p$ is the estimated parameter and $\mbf{v} \in \R^{N\times 1}$ is a noise vector with independent and sub-Gaussian elements\footnote{For simplicity we only consider the real case. The complex case is similar with minor modifications.}. $N$ indicates the number of samples used in the model.\nl
Many real world noise models are sub-Gaussian including Gaussians, finite Gaussian mixtures, all the bounded variables, and any combination of the above. Many real world applications are subject to such noise.\nl
The least squares estimator with $N$ samples is given by
\begin{equation}\label{def:theta_N}
\hat{\mbf{\theta}}^N_{0} = \left(\mbf{A}^T\mbf{A}\right)^{-1}\mbf{A}^T\mbf{x} = \left(\meansum \mbf{a}_n\mbf{a}_n^T\right)^{-1} \meansum \mbf{a}_n^Tx_n
\end{equation}
where $\mbf{a}_n^T, \hsk n=1 \dots N$ are the rows of $\mbf{A}$ and $x_n, \hsk n=\hsk 1 \dots N$ are the data samples.
When $E\left(\mbf{v}\right) = \mbf{0}$, $E\left(\hat{\mbf{\theta}}_0^N\right) = \mbf{\theta}_0$ and the estimator is unbiased.\nl
We want to study the tail distribution of $\norm{\hat{\mbf{\theta}}^N_{0} - \mbf{\theta}_0}_{\infty}$ or more specifically to obtain bounds of the form
\begin{equation}\label{eq:wanted_inequality}
P\left(\norm{\hat{\mbf{\theta}}^N_{0} - \mbf{\theta}_0}_\infty > r\right) < \eps
\end{equation}
as a function of $N$. Furthermore, given $r$, $\eps$ we want to calculate the number of samples needed $N\left(r,\eps\right)$ to achieve the above inequality.
We analyze the case where $\mbf{A}$ is random with bounded elements. \nl
Throughout this paper we use the following mathematical notations:
\begin{defin} \ \\\label{definitions}
\begin{enumerate}
\item{Let $\mbf{B} \in \R^{p\times p}$ be a square matrix; we define the operators $\lambda_{max}\left(\mbf{B}\right)$ and $\lambda_{min}\left(\mbf{B}\right)$ to give the maximal and minimal eigenvalues of $\mbf{B}$ respectively.}
\item{Let $\mbf{C}$ be a matrix. The spectral norm for matrices is given by $\norm{\mbf{C}} \doteq \sqrt{\lambda_{max}\left(\mbf{C^T}\mbf{C}\right)}$.}
\item{A random variable $v$ with $E\left(v\right) = 0$ is called sub-Gaussian if its moment generating function exists and $E\left(\exp\left(sv\right)\right) \leq \exp\left(\frac{s^2R^2}{2}\right)$ ~\cite{vershynin2010introduction}. The minimal $R$ that satisfies this inequality is called the sub-Gaussian parameter of the random variable $v$ and we say that $v$ is sub-Gaussian with parameter $R$.}\label{sub_Gaussian_def}
\end{enumerate}
\end{defin}
\begin{remark}
Assume that $x \sim N\left(0,\sigma^2\right)$, then the moment generating function of $x$ is $M(s) = E\left(\exp\left(sx\right)\right) = \exp\left(\frac{s^2\sigma^2}{2}\right)$. Therefore, by definition \ref{sub_Gaussian_def} $x$ is also sub-Gaussian with parameter $\sigma$.
\end{remark}

\section {Main Result}\label{sec:main_result}
In this section we formulate the main result of this paper, discuss it and provide a proof outline.
\nl
We make the following assumptions regarding the problem. These assumptions are mild and cover a very large set of linear least squares problems. 
\begin{enumerate}[label=\bfseries A\arabic*:]
\item $E\left(v_n\right) = 0 \hsk \forall 1 \leq n \leq N$.\label{assum:1}
\item $P\left(\rank\left({\mbf{A}^T\mbf{A}}\right) = p\right) = 1$
\item $P\left(\abs{a_{ni}} \leq \alpha\right) = 1 \hsk \forall n = 1 \dots N \hsk \forall i = 1 \dots p$
\item For all $N > 0$ there exists $\mbf{M} \in \R^{p\times p}$ such that $\mbf{M} = \frac{1}{N} E\left(\mbf{A}^T\mbf{A}\right)$. We denote $\sigma_{max} \doteq \lambda_{max}\left(\mbf{M}\right)$ and $\sigma_{min} = \lambda_{min}\left(\mbf{M}\right)$.
\item $E\left(v_n | F_{n-1}\right) = 0$. Where $F_{n-1}$ is a flirtation, $v_n$ are independent of $\mbf{A}.$
\item The martingale difference sequence is $\delta$ sub-Gaussian; i.e. $E\left(s v_n | F_{n-1}\right) \leq e^{\frac{s^2 \delta^2}{2}}$.
\end{enumerate}
Assumptions A1-A2 are standard in least squares theory. Assumption A1 assumes that our design is correct. Assumption A2 ensures that the least squares estimator exists. Assumptions A3 and A4 are mild and achievable by normalizing each row of the mixing matrix with the proper scaling of the sub-Gaussian parameter. Assumption A5 means that the noise sequence is a martingale difference sequence and assumption A6 assumes that the noise sequence is sub-Gaussian. Note that the set of assumptions is valid for any type of martingale difference zero mean sub-Gaussian noise model, which is a very wide family of distributions. \nl

The main theorem provides bounds on the convergence rate of the finite sample least squares estimator to the real parameter. The theorem provides the number of samples needed so that the distance between the estimator and the real parameter will be at most $r$ with probability $1-\eps$.

\begin{thm} \emph{\textbf{(Main Theorem)}} \label{thm:main_theorem} \ \\
Let $\mbf{x}$ be defined as in ~\BracketRef{eq:x_definition} and assume assumptions A1-A6. Let $\eps > 0$ and $r > 0$ be given and $\hat{\mbf{\theta}}_0^N$ and $\mbf{\theta}_0$ be defined as previously, then $\forall N > N\left(r,\eps\right)$
\begin{equation}
P\left(\norm{\hat{\mbf{\theta}}_0^N - \mbf{\theta}_0}_{\infty} > r\right) < \eps
\end{equation}
where
\begin{equation}
N\left(r,\eps\right) = \max \left\{N_1\left(r,\eps\right), N_{rand}\left(\eps\right)\right\},
\end{equation}

\begin{equation}
N_1\left(r,\eps\right) = \frac{8\alpha^2\delta^2}{r^2\sigma_{min}^2}\log \frac{2p}{\eps}
\end{equation}
and
\begin{equation}
N_{rand}\left(\eps\right) = \frac{4}{3}\frac{\left(6\sigma_{max} + \sigma_{min}\right)\left(p\alpha^2 + \sigma_{max}\right)}{\sigma_{min}^2}\log \frac{2p}{\eps}.
\end{equation}

\end{thm}

\subsection{Discussion} \
The importance of this result is that it gives an easily calculated bound on the number of samples needed for linear least squares problems. It shows a sharp convergence in probability as a function of $N$, and shows that the number of samples is $O\left(\frac{1}{r^2}\log \frac{1}{\eps}\right)$. Moreover, the result handles the case where the noise is a martingale difference sequence. This is the first finite sample analysis result for least squares under this noise assumption.\nl
The results in this work are given with an $L^{\infty}$ norm. The $L^{\infty}$ results can give confidence bounds for every coordinate of the parameter vector $\mbf{\theta}_0$. Results for other norms can be achieved as well using the relationships between norms.\nl

We start by stating two auxiliary lemmas
\begin{lemma}\label{lemma:general_least_squares_calculation}
Let $\mbf{x}$ be defined as in \BracketRef{eq:x_definition}. Assume A1-A6 hold. Furthermore, let $\hat{\mbf{\theta}}^N_{0}$ be defined in \BracketRef{def:theta_N} and let $r > 0$ be given, then
\begin{equation}
\begin{array}{lcl}
&P\left(\abs{\left(\hat{\mbf{\theta}}^N_0 - \mbf{\theta}_0\right)_i} > r\right)&\\
\leq& P\left(\abs{\meansum a_{ni}v_n} > \frac{r}{\lambda_{max}\left(\left(\frac{1}{N}\mbf{A}^T\mbf{A}\right)^{-1}\right)}\right).&
\end{array}
\end{equation}
\end{lemma}
\begin{proof}
This lemma can be proven by a straightforward computation and the proof is left to the reader.
\end{proof}

\begin{lemma}\label{lemma:omega_bound}
Under assumptions A2-A4 and for all $N \geq N_{rand}\left(\eps'\right)$
\begin{equation}
P\left(\lambda_{max}\left(\frac{1}{N}\left(\mbf{A}^T\mbf{A}\right)\right)^{-1} \geq \frac{2}{\sigma_{min}}\right) \leq \eps',
\end{equation}
where
\begin{equation}
N_{rand}\left(\eps'\right) = \frac{4}{3}\frac{\left(6\sigma_{max} + \sigma_{min}\right)\left(p\alpha^2 + \sigma_{max}\right)}{\sigma_{min}^2}\log\left(\frac{p}{\eps'}\right).
\end{equation}
\end{lemma}
\begin{proof}
The proof is a generalization of the proof of theorem 4.1 in ~\cite{krikheli2016finite}. It is omitted due to space limitations.
\end{proof}

\subsection{Prof Outline} \
\begin{proof}
We wish to study the term
\begin{equation}
P\left(\norm{\hat{\mbf{\theta}}_0^N - \mbf{\theta}_0}_{\infty} > r\right).
\end{equation}
In order to do so we start by bounding each of the terms in the vector separately and use a union bound approach to achieve the $L_{\infty}$ bound. We start by analyzing the term
\begin{equation}
P\left(\abs{\left(\hat{\mbf{\theta}}_0^N - \mbf{\theta}_0\right)_i} > r\right).
\end{equation}
Using lemma ~\ref{lemma:general_least_squares_calculation} we achieve
\begin{equation}
\begin{array}{lcl}
& P\left(\abs{\left(\hat{\mbf{\theta}}_0^N - \mbf{\theta}_0\right)_i} > r\right) &\\
\leq& P\left(\abs{\meansum a_{ni}v_n} > \frac{r}{\lambda_{max}\left(\frac{1}{N}\mbf{A}^T\mbf{A}\right)^{-1}}\right).&
\end{array}
\end{equation}
We define the set of events 
\begin{equation}\label{def_psi_1_mart}
\Psi_1 \doteq \left\{\mbf{X} : \lambda_{max}\left(\frac{1}{N}\mbf{A}^T\mbf{A}\right)^{-1} \geq \frac{2}{\sigma_{min}}\right\}.
\end{equation}
We want to study the number of samples required to achieve that $P\left(\mbf{X} \in \Psi_1\right) \leq \frac{\eps}{2}$.
In order to achieve this, we use lemma \ref{lemma:omega_bound} with parameter $\eps' = \frac{\eps}{2}$ to find that $\forall N > N_{rand}\left(\eps\right)$
\begin{equation}
P\left(\mbf{X} \in \Psi_1\right) = P\left(\lambda_{max}\left(\frac{1}{N}\mbf{A}^T\mbf{A}\right)^{-1} \geq \frac{2}{\sigma_{min}}\right) \leq \frac{\eps}{2}.
\end{equation}
We denote 
\begin{defin}
$\mbf{c_i}$ the $i$-th column of $\mbf{A}$.
\end{defin}
We now assume that $\mbf{X} \notin \Psi_1$. 
Under this assumption the following inequality holds
\begin{equation}
P\left(\abs{\left(\hat{\mbf{\theta}}^N_0 - \mbf{\theta}_0\right)_0} > r\right)
\leq P\left(\frac{1}{N} \mbf{c}_i^T\mbf{v} > \frac{r\sigma_{min}}{2}\right).
\end{equation}
We denote by
\begin{equation}\label{def_psi_2_mart}
\Psi_2\left(i\right) \doteq \left\{\mbf{X} : \frac{1}{N} \mbf{c}_i^T\mbf{v} > \frac{r\sigma_{min}}{2}\right\}.
\end{equation}
We now obtain a bound on the number of samples required to ensure that
\begin{equation}
P\left(\mbf{X} \in \Psi_2\left(i\right)\right) = P\left(\frac{1}{N} \mbf{c}_i^T\mbf{v} > \frac{r\sigma_{min}}{2}\right) \leq \frac{\eps}{2p}.
\end{equation}

We now outline a proof for a concentration result for a sub-Gaussian martingale difference sequence using similar methods to ~\cite{thoppe2015concentration}. We begin by bounding the moment generating function. We start by bounding $E\left(\exp\left(s\mbf{c}_i^T\mbf{v}\right)\right)$ and then we use Markov's inequality. Applying assumption A3 and A6 we obtain
\begin{equation}
E\left(\exp\left(s\mbf{c}_i^T\mbf{v}\right)\right)
\leq E\left(\exp\left(s\alpha \dispsum_{n=1}^{N-1}v_n\right)\right) e^{\frac{s^2\alpha^2\delta^2}{2}}.
\end{equation}
Iterating this procedure yields
\begin{equation}\label{eq:sub_Gaussian_expectation_bound}
E\left(\exp\left(s\mbf{c}_i^T\mbf{v}\right)\right) \leq \exp\left(\frac{Ns^2\alpha^2\delta^2}{2}\right).
\end{equation}
Looking now at the original equation we use the Laplace method and Markov's inequality alongside equation ~\BracketRef{eq:sub_Gaussian_expectation_bound} to achieve
\begin{equation}
P\left(\mbf{X} \in \Psi_2\left(i\right)\right) \leq
\exp\left(\frac{N}{2}\left(s^2\alpha^2\delta^2 - sr\sigma_{min}\right)\right)
\end{equation}
Optimizing over $s > 0$ we achieve
\begin{equation}
P\left(\mbf{c}_i^T\mbf{v} > \frac{Nr\sigma_{min}}{2}\right) \leq \exp\left(-\frac{Nr^2\sigma_{min}^2}{8\alpha^2\delta^2}\right).
\end{equation}
Choosing $N > \frac{8\alpha^2\delta^2}{r^2\sigma_{min}^2}\log \frac{2p}{\eps}$ ensures that
\begin{equation}
P\left(\mbf{X} \in \Psi_2\left(i\right)\right) = P\left(\frac{1}{N} \mbf{c}_i^T\mbf{v} > \frac{r\sigma_{min}}{2}\right) < \frac{\eps}{2p}.
\end{equation}
We achieved a bound for each coordinate separately. The last step of the proof is to use the union bound on these terms to achieve a bound on the $L_{\infty}$ norm of the vector. We define
\begin{equation}
\Psi_2 \doteq \bigcup_{i=1}^p \Psi_{2}\left(i\right).
\end{equation}
Using the union bound we obtain $\forall N > N\left(r,\eps\right)$
\begin{equation}
P\left(\mbf{X} \in \Psi_2\right) \leq \dispsum_{i=1}^p P\left(\mbf{X} \in \Psi_2\left(i\right)\right) \leq \frac{\eps}{2}.
\end{equation}
Using the union bound again we obtain $P\left(\mbf{X} \in \Psi_1 \cup \Psi_2\right) \leq \eps$.
This completes the proof.
\end{proof}

\section{Simulation Results}

Assume that we have a linear system with unknown parameters with noise that is filtered using a finite impulse response system and a sub-Gaussian signal. We can write
\begin{equation}
x_n = \mbf{a}_n^T \mbf{\theta} + \dispsum_{i=0}^k j\left(i\right)H\left(n-i\right) + w\left(n\right).
\end{equation}
where $j\left(i\right)$ is an i.i.d zero mean bounded signal for example a BPSK signal. We denote $\eta$ as the bound for $j\left(i\right)$, i.e. $P\left(j\left(i\right) \leq \eta\right) = 1$. We also assume that $H\left(n\right)$ is an unknown system. 
We now prove that the noise sequence $v_n = \dispsum_{i=0}^k j\left(i\right)H\left(n-i\right) + w\left(n\right)$ is a zero mean martingale difference, that it is sub-Gaussian and thus admits assumptions A5 and A6. If so, we can use theorem \ref{thm:main_theorem} to calculate the number of samples required to achieve a certain finite sample performance for this interesting model.
\begin{multline}
E\left(v_n | F_{n-1}\right) = E\left(\dispsum_{i = 0}^k j\left(i\right)H\left(n-i\right) + w\left(n\right)  | F_{n-1}\right) \\ 
= E\left(\dispsum_{i=0}^k j\left(i\right) H \left(n-i\right) | F_{n-1}\right) + E\left(w_n | F_{n-1}\right)\\
 = E\left(\dispsum_{i=0}^k j\left(i\right) H \left(n-i\right) | F_{n-1}\right) = 0.
\end{multline}
The second equatlity follows from the independence of the random variables $w\left(n\right)$, $j\left(n\right)$ and $H\left(n\right)$. The next equality follows from the fact that $w\left(n\right)$ is zero mean. The last equality follows from the fact that $E\left(j\left(n\right)\right) = 0$. We now prove that the $v_n$ is sub-Gaussian. We use the assumption that $j\left(n\right) \leq \eta$ and that $H\left(n\right)$ and $w\left(n\right)$ are sub-Gaussian with parameter $R_1$ and $R_2$ respectively. Using these facts with the property that $j\left(k\right)H\left(n\right)$ is sub-Gaussian as they are independent and $j\left(i\right) \leq \eta$ and the fact that linear combinations of sub-Gaussian random variables is sub-Gaussian ~\cite{vershynin2010introduction} we can conclude that $v_n$ is sub-Gaussian and admits assumption A6. We have now proven that this example admits all the assumptions of theorem \ref{thm:main_theorem} and therefore we can use the theorem to bound the number of samples needed to achieve a predefined performance. Fig 1. shows the performance of the bound in this interesting case. We see that while the bound is not tight, the overall performance is similar. This example demonstrates the strengths of the results in this paper. Many signal processing applications such as this example can be analyzed using our results.
\begin{figure}
    \centering
    \includegraphics[scale=0.4]{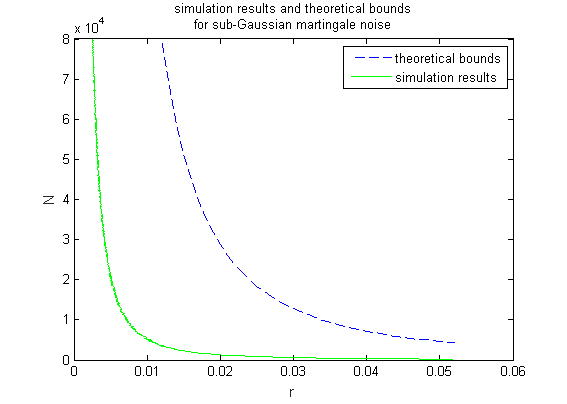}
    \caption{Simulation results and martingale difference theorem bounds for a sub-Gaussian martingale difference sequence noise with $\eps = 0.2, \sigma_{min} = \sigma_{max} = 10$, $\delta = 4$ and $p = 2$. The graph is for $N$ as a function of $r$.}
\end{figure}

\section{Concluding remarks}\label{sec:concluding_remarks}
In this paper we examined the finite sample performance of the $L^{\infty}$ error of the linear least squares estimator. We showed very fast convergence of the number of samples required as a function of the probability of the $L^{\infty}$ error. We showed that the number of samples required to achieve a maximal deviation $r$ with probability $1-\eps$ is $N \sim O\left(\frac{1}{r^2}\log \frac{1}{\eps}\right)$. The main theorem deals with least squares in very general noise models; therefore the bounds may be important in many interesting applications. We used simulations to demonstrate the results. Our simulation results suggest that the bounds given in this paper have similar properties as the simulation results. We showed that the interesting example of a finite impulse response filtered interference with sub-Gaussian noise model can be modeled as a sub-Gaussian martingale difference model in our setup and our theorem can give bounds on the number of samples required to achieve the required performance in this important case. This result has significant implications for the analysis of least squares problems in communications and signal processing. We also believe that the Sub-Gaussian parameter can be replaced with bounds on a few moments of the distribution and can relax the bounds. This is left for further study.

\newpage
\bibliographystyle{IEEEtran}
\bibliography{icassp_2018_least_squares_with_martingale_difference_noise}

\begin{thebibliography}{10}
\providecommand{\url}[1]{#1}
\csname url@samestyle\endcsname
\providecommand{\newblock}{\relax}
\providecommand{\bibinfo}[2]{#2}
\providecommand{\BIBentrySTDinterwordspacing}{\spaceskip=0pt\relax}
\providecommand{\BIBentryALTinterwordstretchfactor}{4}
\providecommand{\BIBentryALTinterwordspacing}{\spaceskip=\fontdimen2\font plus
\BIBentryALTinterwordstretchfactor\fontdimen3\font minus
  \fontdimen4\font\relax}
\providecommand{\BIBforeignlanguage}[2]{{%
\expandafter\ifx\csname l@#1\endcsname\relax
\typeout{** WARNING: IEEEtran.bst: No hyphenation pattern has been}%
\typeout{** loaded for the language `#1'. Using the pattern for}%
\typeout{** the default language instead.}%
\else
\language=\csname l@#1\endcsname
\fi
#2}}
\providecommand{\BIBdecl}{\relax}
\BIBdecl

\bibitem{zhang2008image}
X.~Zhang and X.~Wu, ``Image interpolation by adaptive 2-{D} autoregressive
  modeling and soft-decision estimation,'' \emph{IEEE Transactions on Image
  Processing}, vol.~17, no.~6, pp. 887--896, June 2008.

\bibitem{hung2012robust}
K.~W. Hung and W.~C. Siu, ``Robust soft-decision interpolation using weighted
  least squares,'' \emph{IEEE Transactions on Image Processing}, vol.~21,
  no.~3, pp. 1061--1069, March 2012.

\bibitem{so2011linear}
H.~C. So and L.~Lin, ``Linear least squares approach for accurate received
  signal strength based source localization,'' \emph{IEEE Transactions on
  Signal Processing}, vol.~59, no.~8, pp. 4035--4040, Aug 2011.

\bibitem{veraart2013weighted}
J.~Veraart, J.~Sijbers, S.~Sunaert, A.~Leemans, and B.~Jeurissen, ``Weighted
  linear least squares estimation of diffusion {MRI} parameters: strengths,
  limitations, and pitfalls,'' \emph{Neuroimage}, vol.~81, pp. 335--346, 2013.

\bibitem{leshem1999direction}
A.~Leshem and A.~J. van~der Veen, ``Direction-of-arrival estimation for
  constant modulus signals,'' \emph{IEEE Transactions on Signal Processing},
  vol.~47, no.~11, pp. 3125--3129, Nov 1999.

\bibitem{stoica1989music}
P.~Stoica and A.~Nehorai, ``{MUSIC}, maximum likelihood and {C}ramer-{R}ao
  bound,'' in \emph{ICASSP-88., International Conference on Acoustics, Speech,
  and Signal Processing}, Apr 1988, pp. 2296--2299 vol.4.

\bibitem{djuric2002sequential}
P.~M. Djuri{\'c}, J.~H. Kotecha, F.~Esteve, and E.~Perret, ``Sequential
  parameter estimation of time-varying non-{G}aussian autoregressive
  processes,'' \emph{EURASIP Journal on Applied Signal Processing}, vol. 2002,
  no.~1, pp. 865--875, 2002.

\bibitem{banerjee2014performance}
S.~Banerjee and M.~Agrawal, ``On the performance of underwater communication
  system in noise with {G}aussian mixture statistics,'' in \emph{2014 Twentieth
  National Conference on Communications (NCC)}, Feb 2014, pp. 1--6.

\bibitem{tan2014wiener}
J.~Tan, D.~Baron, and L.~Dai, ``Wiener filters in {G}aussian mixture signal
  estimation with $\ell _\infty$ -norm error,'' \emph{IEEE Transactions on
  Information Theory}, vol.~60, no.~10, pp. 6626--6635, Oct 2014.

\bibitem{bhatia2007non}
V.~Bhatia and B.~Mulgrew, ``Non-parametric likelihood based channel estimator
  for {G}aussian mixture noise,'' \emph{Signal Processing}, vol.~87, no.~11,
  pp. 2569--2586, 2007.

\bibitem{wang1999robust}
X.~Wang and H.~V. Poor, ``Robust multiuser detection in non-{G}aussian
  channels,'' \emph{IEEE Transactions on Signal Processing}, vol.~47, no.~2,
  pp. 289--305, Feb 1999.

\bibitem{engle1982autoregressive}
R.~F. Engle, ``Autoregressive conditional heteroscedasticity with estimates of
  the variance of united kingdom inflation,'' \emph{Econometrica: Journal of
  the Econometric Society}, pp. 987--1007, 1982.

\bibitem{lai1982least}
T.~L. Lai and C.~Z. Wei, ``Least squares estimates in stochastic regression
  models with applications to identification and control of dynamic systems,''
  \emph{The Annals of Statistics}, pp. 154--166, 1982.

\bibitem{lai1983asymptotic}
T.~Lai and C.~Wei, ``Asymptotic properties of general autoregressive models and
  strong consistency of least-squares estimates of their parameters,''
  \emph{Journal of multivariate analysis}, vol.~13, no.~1, pp. 1--23, 1983.

\bibitem{nelson1980note}
P.~I. Nelson, ``A note on strong consistency of least squares estimators in
  regression models with martingale difference errors,'' \emph{The Annals of
  Statistics}, pp. 1057--1064, 1980.

\bibitem{christopeit1980strong}
N.~Christopeit and K.~Helmes, ``Strong consistency of least squares estimators
  in linear regression models,'' \emph{The Annals of Statistics}, pp. 778--788,
  1980.

\bibitem{kramer1980finite}
W.~Kr{\"a}mer, ``Finite sample efficiency of ordinary least squares in the
  linear regression model with autocorrelated errors,'' \emph{Journal of the
  American Statistical Association}, vol.~75, no. 372, pp. 1005--1009, 1980.

\bibitem{lai1985strong}
T.~L. Lai, H.~Robbins, and C.~Z. Wei, ``Strong consistency of least squares
  estimates in multiple regression,'' in \emph{Herbert Robbins Selected
  Papers}.\hskip 1em plus 0.5em minus 0.4em\relax Springer, 1985, pp. 510--512.

\bibitem{gleser1965asymptotic}
L.~J. Gleser, ``On the asymptotic theory of fixed-size sequential confidence
  bounds for linear regression parameters,'' \emph{The Annals of Mathematical
  Statistics}, pp. 463--467, 1965.

\bibitem{oliveira2013lower}
R.~I. Oliveira, ``The lower tail of random quadratic forms, with applications
  to ordinary least squares and restricted eigenvalue properties,'' \emph{arXiv
  preprint arXiv:1312.2903}, 2013.

\bibitem{hsu2014random}
D.~Hsu, S.~M. Kakade, and T.~Zhang, ``Random design analysis of ridge
  regression,'' \emph{Foundations of Computational Mathematics}, vol.~14,
  no.~3, pp. 569--600, 2014.

\bibitem{audibert2010robust}
J.-Y. Audibert and O.~Catoni, ``Robust linear regression through
  {PAC}-{B}ayesian truncation,'' \emph{Preprint, URL http://arxiv.
  org/abs/1010.0072}, vol.~38, p.~60, 2010.

\bibitem{audibert2011robust}
------, ``Robust linear least squares regression,'' \emph{The Annals of
  Statistics}, pp. 2766--2794, 2011.

\bibitem{krikheli2016finite}
M.~Krikheli and A.~Leshem, ``Finite sample performance of least squares
  estimation in sub-{G}aussian noise,'' in \emph{2016 IEEE Statistical Signal
  Processing Workshop (SSP)}, June 2016, pp. 1--5.

\bibitem{vershynin2010introduction}
R.~Vershynin, ``Introduction to the non-asymptotic analysis of random
  matrices,'' \emph{arXiv preprint arXiv:1011.3027}, 2010.

\bibitem{thoppe2015concentration}
G.~Thoppe and V.~S. Borkar, ``A concentration bound for stochastic
  approximation via {A}lekseev's formula,'' \emph{arXiv preprint
  arXiv:1506.08657}, 2015.

\end{thebibliography}

\end{document}